\documentclass[11pt]{llncs}
\usepackage{amsmath}
\usepackage{algorithm}
\usepackage[noend]{algpseudocode}

\makeatletter
\def\BState{\State\hskip-\ALG@thistlm}
\makeatother

\def\CC{\mathcal{C} }
\def\DD{\mathcal{D}}
\def\HH{\mathcal{H}}
\def\RR{\mathcal{R}}

\begin{document}

\title{\bf On Algorithms for $L$-bounded Cut Problem\thanks{This research
was partially supported by project GA15-11559S of GA \v{C}R.}
}
\author{Petr Kolman}
\institute{Department of Applied Mathematics\\ 
	Faculty of Mathematics and Physics\\ 
Charles University, Prague, Czech Republic\\
\email{kolman@kam.mff.cuni.cz} }


\maketitle
\begin{abstract}
Given a graph $G=(V,E)$ with two distinguished vertices $s,t\in V$ and an
integer parameter $L>0$, an {\em $L$-bounded cut} is a subset $F$ of edges
(vertices) such that the every path between $s$ and $t$ in $G\setminus F$ has
length more than $L$. The task is to find an $L$-bounded cut of minimum
cardinality. 

Though the problem is very simple to state and has been studied since the
beginning of the 70's, it is not much understood yet. The problem is known to
be $\cal{NP}$-hard to approximate within a small constant factor even for
$L\geq 4$ (for $L\geq 5$ for the vertex cuts).  
On the other hand, the best known
approximation algorithm for general graphs has approximation ratio only
$\mathcal{O}({n^{2/3}})$ in the edge case, and $\mathcal{O}({\sqrt{n}})$ in the
vertex case, where $n$ denotes the number of vertices. 

We show that for planar graphs, it is possible to solve both the edge- and the
vertex-version of the problem optimally in time $\mathcal{O}(L^{3L}n)$.  That is,
the problem is fixed parameter tractable (FPT) with respect to $L$ on planar
graphs. Furthermore, we show that the problem remains FPT even for bounded
genus graphs, a super class of planar graphs.

Our second contribution deals with approximations of the vertex version of the
problem. We describe an algorithm that for a given a graph $G$, its tree
decomposition of treewidth $\tau$ and vertices $s$ and $t$ computes a
$\tau$-approximation of the minimum $L$-bounded $s-t$ vertex cut;  
if the decomposition is not given, then the approximation ratio
is $\mathcal{O}(\tau \sqrt{\log \tau})$.
For graphs with treewidth bounded by $\mathcal{O}(n^{1/2-\epsilon})$ for any
$\epsilon>0$, but not by a constant, this is the best approximation in terms
of~$n$ that we are aware of.
\end{abstract}

\section{Introduction}

The subject of this paper is a variation of the classical $s$-$t$ cut problem,
namely the 
{\em minimum $L$-bounded edge (vertex) cut problem}: given a graph
$G=(V,E)$ with two distinguished vertices $s,t\in V$ and an integer parameter
$L>0$, find a subset $F$ of edges (vertices) of minimum cardinality such that 
every path between $s$ and $t$ in $G\setminus F$ has length more than $L$.
The problem has been studied in various contexts since the beginning of the
70's (e.g.,~\cite{AK:71,LNP:78,BEH:10})
and occasionally it appears also under the name the {\em short paths interdiction 
problem}~\cite{KBB:08}. 

Closely related is another $\cal{NP}$-hard problem, namely the {\em shortest
path most vital edges and vertices problem} (e.g. \cite{BGV:89,BKS:95,BNN:15}):
given a graph $G$, two distinguished vertices $s$ and $t$ and an integer $k$,
the task is to find a subset $F$ of $k$ edges (vertices) whose removal
maximizes the increase in the length of the shortest path between $s$ and $t$.
If we introduce an additional parameter -- the desired minimum distance of $s$
and $t$ -- we obtain a parameterized version of the $L$-bounded cut problem:
given a graph $G$, two distinguished vertices $s$ and $t$ and integers $k$ and
$L$, does there exist a subset $F$ of at most $k$ edges (vertices) such that
every path between $s$ and $t$ in $G\setminus F$ has length more than $L$?  We
also note that $\cal{NP}$-hardness of the shortest path most vital edges
(vertices) problem immediately implies $\cal{NP}$-hardness of the $L$-bounded
edge (vertex) cut problem, and vice versa.

In contrast to many other cut problems on graphs (e.g., multiway cut, multicut,
sparsest cut, balanced cut, maximum cut, multiroute cut), the known approximations of
the minimum $L$-bounded cut problem are substantially weaker. In this work
we focus on algorithms for restricted graph classes, namely planar graphs,
bounded genus graphs
and graphs with bounded, yet not constant, treewidth,
and provide new results for the $L$-bounded cut problem on them; the results
for planar graphs solve one of the open problems suggested by 
Bazgan et al.~\cite{BNN:15}.
We also remark that the generic approximation scheme of
Czumaj~et.~al~\cite{CHLN:05} for $\cal{NP}$-hard problems in graphs with
superlogarithmic treewidth is not applicable for the $L$-bounded cut problem.

\paragraph{Related Results.}
$\cal{NP}$-hardness of the shortest path most vital edges problem (and, thus,
as noted above, also of the $L$-bounded cut problem) was proved by Bar-Noy et
al.~\cite{BKS:95}.  The best known approximation algorithm for the minimum
$L$-bounded cut problem on general graphs has approximation ratio only
$\mathcal{O}({\min\{L,n/L\}}) \subseteq \mathcal{O}({\sqrt{n}})$ for the vertex
case and $\mathcal{O}({\min\{L,n^2/L^2,\sqrt{m}\})} \subseteq
\mathcal{O}({n^{2/3}})$ for the edge case, where $m$ denotes the number of
edges and $n$ the number of vertices~\cite{BEH:10}.  On the lower bound side,
the edge version of the problem is known to be $\cal{NP}$-hard to approximate
within a factor of $1.1377$ for $L\geq 4$, and the vertex version for $L\geq
5$~\cite{BEH:10}; for smaller values of $L$ the problem is solvable in
polynomial time~\cite{LNP:78,MM:10}.  Independently, Khachiyan et
al.~\cite{KBB:08} proved that a version of the problem with edge lengths is
$\cal{NP}$-hard to approximate within a factor smaller than $1.36$.  Recently,
assuming the Unique Games Conjecture, Lee proved that the problem is
$\cal{NP}$-hard to approximate within {\em any} constant factor~\cite{Lee:17}.

An instance of the $L$-bounded cut problem on a graph of treewidth $\tau$  can
be cast as an instance of constraint satisfaction problem (CSP) with domain of
size $L+1$ and treewidth $\tau$\footnote{For the sake of completeness, in
Appendix A we provide details about this reduction.}. As CSP instances with
treewidth bounded by $\tau$ and domain by $D$ can be solved in time $O(D^\tau
n)$~\cite{Freuder:90} (when a tree decomposition of treewidth $\tau$ of the
constraint graph is given), the problem is fixed parameter tractable with
respect to $L$ and $\tau$. Dvo\v r\'ak and Knop~\cite{DK:15} provide a direct
proof of the same result with a slightly worse dependance on $L$ and $\tau$;
they also prove that the problem is $W[1]$-hard when parameterized by the
treewidth only. 

From the point of view of parameterized complexity, the problem was also studied
by Golovach and Thilikos~\cite{GT:11}, Bazgan et al.~\cite{BNN:15} 
and by Fluschnik et al.~\cite{FHNN:16}.

For planar graphs, the problem is known to be $\cal{NP}$-hard~\cite{FHNN:15,PS:16},
too, and the problem has no polynomial-size kernel when parameterized by the
combination of $L$ and the size of the optimal solution~\cite{FHNN:16}. 

For more detailed overview of other related results and applications, 
we refer to the papers~\cite{KBB:08,BEH:10,MM:10}.

\paragraph{Our Contribution.}
We show that on planar graphs, both the edge- and the vertex-version of the
problem are solvable in time $\mathcal{O}(L^{3L}n)$. That is, we show that on
planar graphs the minimum $L$-bounded cut problem is fixed parameter tractable (FPT)
with respect to $L$. 
Furthermore, we show that the problem remains FPT even for
bounded genus graphs, a super class of planar graphs.
This is in contrast with the situation for general graphs 
-- the problem is NP-hard even for $L=4$ and $L=5$, for the edge- and vertex-versions,
respectively. 

Our second contribution is a $\tau$-approximation algorithm for the vertex
version of the problem, if a tree decomposition of width $\tau$ is given.  If
the decomposition is not given, then using the currently best known tree
decomposition algorithms, we obtain an $\mathcal{O}(\tau\sqrt{\log
\tau})$-approximation for general graphs with treewidth $\tau$, and an
$\mathcal{O}(\tau)$-approximation for planar graphs, graphs excluding a fixed
minor and graphs with treewidth bounded by $\mathcal{O}(\log n)$.  For graphs
with treewidth bounded by $\tau=\mathcal{O}(n^{1/2-\epsilon})$ for any
$\epsilon>0$, but not by a constant, in terms of $n$, this is the best
approximation we are aware of.

Our results are based on a combination of observations about the structure
of $L$-bounded cuts and various known results. The proofs are straightforward
but apparently non-obvious, considering the attention given to the problem in
recent years.

\section{$L$-Bounded Cut is FPT on Planar Graphs and Bounded Genus Graphs}

Throughout the paper, given a graph $G=(V,E)$ and $u,v\in V$,
we use $d(u,v)$ to denote the shortest path distance between $u$ and $v$,
that is, the number of edges on the shortest path.

Our main tools are the following two well-known results.
\begin{theorem}[Robertson and Seymour~\cite{RS:84},
Bodlaender~\cite{Bodlaender:88}]\label{thm:treewidth} The treewidth of a planar
graph with radius $d$ is at most $3d$.  
\end{theorem}

\begin{theorem}[Freuder~\cite{Freuder:90}]\label{thm:CSP}
CSP instances with treewidth bounded by $\tau$, domain by $D$ and size by $n$ are
solvable in time $O(D^{\tau} n)$.
\end{theorem}
Since the $L$-bounded cut problem on a graph of treewidth $\tau$ can be cast as
a CSP instance with treewidth $\tau$, domain of size $L+1$ and size $n$,
the minimum $L$-bounded cut problem is solvable in time $\mathcal{O}(L^{\tau}
n)$, as already stated in the introduction and explained in Appendix A.

The main result of this section says that the $L$-bounded cut problem on planar
graphs is fixed parameter tractable, with respect to the parameter $L$.
\begin{theorem}\label{thm:planarFPT}
The minimum $L$-bounded edge (vertex) cut problem on planar graphs is solvable in time
$O(L^{3L} n)$.
\end{theorem}

\begin{proof}
We prove the theorem for the edge version; the proof for the vertex version is
analogous.

Given a graph $G=(V,E)$, $s,t\in V$ and an integer $L$, let $V'=\{v\in V\ |\
d(s,v)+d(t,v)\leq L\}$. In words, $V'$ is the subset of vertices lying on paths of
length at most $L$ between $s$ and $t$. 
Without loss of generality we assume that $d(s,t)\leq L$ -- otherwise the problem is trivial.
Let $G'$ be the subgraph of $G$ induced
by $V'$.  Note that the radius of $G'$ is at most $L$ as, by definition,
$d(s,v)\leq L$ for every $v\in V'$.

The set $V'$ (and, thus, the subgraph $G'$) can be computed using a linear time
algorithm for single-source shortest paths~\cite{KRRS:94} that we run twice,
once for $s$ and once for $t$. Note that both $s$ and $t$ belong to $V'$.

Obviously, $G'$ is a planar graph, and by Theorem~\ref{thm:treewidth}, its 
treewidth is at most $3L$. We solve the
$L$-bounded problem for $G'$ and $s$ and $t$ by Theorem~\ref{thm:CSP}. Let $F\subseteq E(G')$ be the
optimal solution. We claim that $F$ is an optimal solution for the original instance
of the problem on $G$ as well. To prove {\em feasibility} of $F$, assume, 
for contradiction, that there exists
an $s-t$-path $p$ of length at most $L$ in $(V,E\setminus F)$. As there is no
such path in $G'\setminus F$, $p$ has to use at least one vertex $v$ from $V\setminus V'$.
However, this yields a contradiction: on one hand, $d(s,v)+d(v,t)\leq L$ as $v$ is
on an $s-t$-path of length at most $L$, on the other hand, $d(s,v)+d(v,t)>L$
as $v$ is not in $V'$. 
Concerning the {\em optimality} of $F$, it is sufficient to note that the size of
an optimal solution for the subgraph $G'$ is a lower bound on the size of an
optimal solution for $G$.
\qed
\end{proof}

Theorem~\ref{thm:treewidth} was generalized by Eppstein~\cite{Eppstein:00}
to graphs of bounded genus and this result makes it possible to generalize
Theorem~\ref{thm:planarFPT} also to graphs of bounded genus.
\begin{theorem}[Eppstein~\cite{Eppstein:00}]\label{thm:bounded_genus}
There exists a constant $\hat{c}$ such that 
the treewidth of every graph with radius $d$ and genus $g$ is at most $\hat{c} d g$.
\end{theorem}

In the same way as we used Theorem~\ref{thm:treewidth} to prove 
in Theorem~\ref{thm:planarFPT} 
fixed parameter tractability of the $L$-bounded cut for planar graphs, we can
use Theorem~\ref{thm:bounded_genus} to prove fixed parameter tractability
of the $L$-bounded cut problem on graphs of bounded genus.
Thus, we obtain the following theorem.

\begin{theorem}\label{thm:boundedgenusFPT}
The minimum $L$-bounded edge (vertex) cut problem on graphs with genus $g$ 
is solvable in time $O(L^{\hat{c} g L} n)$.
\end{theorem}


\section{$\tau$-approximation for $L$-bounded Vertex Cuts}


In this section we describe an algorithm for the $L$-bounded vertex cut problem
whose approximation ratio is parametrized by the the treewidth of (a given tree
decomposition of) the underlying graph.  For notions related to the treewidth
of a graph and tree decomposition we stick to the standard terminology as given
in the book by Kloks~\cite{Kloks:94}. To distinguish vertices of 
a graph and of its tree decomposition, we call the vertices
of the tree decomposition {\em nodes}.

Given a graph $G=(V,E)$ with two distinguished vertices $s, t\in V$, an integer
parameter $L > 0$ and a rooted tree decomposition $T$ of treewidth
$\tau$ of the graph $G$, the $L$-bounded
cut is computed using the recursive procedure $L$-bounded\_cut$(G,T,s,t,L)$
described bellow.
Throughout this section we assume that the vertices $s$ and $t$ are not
connected by an edge, as otherwise there is no $L$-bounded $s-t$
vertex cut in $G$.

For a node $a$ of the tree decomposition $T$, we
use $B(a)\subseteq V$ to denote the bag of the node $a$. 
In the description of the algorithm, 
mincut$(G,s,t)$ is a procedure that returns a minimum vertex
$s-t$ cut in $G$; given in addition a vertex set
$C\subset V$, prune$(G,T,C)$ is a procedure that
deletes from $G$ the vertices in $C$ and modifies the decomposition tree $T$ by
deleting these vertices from all bags.  By $d(G,s,t)$ we denote the distance
between $s$ and $t$ in $G$; given a connected subtree $R$ of a rooted tree
$T$, a {\em deepest node} in $R$ is a node in $R$ with no child in $R$.
Given a node $b$ of the rooted tree decomposition $T$, we denote by $T_b$
the subtree of $T$ consisting of $b$ and of all its descendants,
and by $G_b$ the subgraph of $G$ induced by vertices in bags of $T_b$;
similarly, we denote by $\bar T_b$ the subtree of $T$ consisting of
all nodes in $T$ except for the descendants of $b$ and by $\bar G_b$
the subgraph of $G$ induced by vertices in bags of $\bar T_b$.

\begin{algorithm}
\caption{$L$-bounded\_cut$(G,T,s,t,L)$}
\begin{algorithmic}[1]
\If {$d(G,s,t)>L$} \textbf{return} $\emptyset$
\ElsIf{$\forall a\in V(T): \ |B(a)\cap \{s,t\}|\leq 1$} \textbf{return} mincut$(G,s,t)$
\Else \ $R\gets \{a\in V(T)\ | \ s,t\in B(a) \text{ and } d({G_a},s,t)\leq L\}$
\State $b \gets$ a deepest node in $R$
\If {$B(b)=\{s,t\}$} 
\State $S_1 \gets$ $L$-bounded\_cut$(G_b,T_b,s,t,L)$
\State $S_2 \gets$ $L$-bounded\_cut$(\bar G_b,\bar T_b,s,t,L)$
\State \textbf{return} $S_1 \cup S_2$
\Else
\EndIf
\State $(G',T') \gets$ prune$(G,T,B(b)\setminus\{s,t\})$  
\State $S' \gets$ $L$-bounded\_cut$(G',T',s,t,L)$
\State \textbf{return} $S'\cup B(b)\setminus\{s,t\}$
\EndIf
\end{algorithmic}
\end{algorithm}

\begin{theorem}\label{thm:approx}
Given a graph $G$, its rooted tree decomposition $T$ of treewidth $\tau$,
vertices $s$ and $t$ and an integer $L$,
the above algorithm finds in polynomial time
a $\tau$-approximation
of the minimum $L$-bounded vertex cut.
\end{theorem}
\begin{proof}
For notational simplicity, we use the term an {\em $L$-bounded path}
for a path of length at most $L$, and we use $T^s$ to denote the
subtree of $T$ induced by nodes with $s$ in their bag, that is,
induced by the set $\{a\in V(T)\ | \ s\in B(a)\}$; the subtree
$T^t$ is defined analogously.

We start by showing the {\em correctness} of the algorithm:
we claim that the algorithm finds a feasible solution for the $L$-bounded
vertex cut problem. The proof is by induction on the recursion depth.
For the final recursive calls, dealt with in steps 1 and 2, the claim
is obvious from the description of the algorithm.

{\em Inductive step.}
Consider a run of the procedure with a graph $G$ and its tree decomposition
$T$, and let $R$ and  $b$ be the objects defined by the procedure in steps 3
and 4. We distinguish two cases (step 5): 
i) the bag $B(b)$ contains the two vertices $s$ and
$t$ only, and, 
ii) the bag $B(b)$ contains in addition to $s$ and $t$ at least one
other vertex. 

In the first case, realizing that $\{s,t\}$ is a separator in $G$, we note
that any $L$-bounded $s-t$ path in $G$ either belongs entirely to $G_b$ or 
belongs entirely to $\bar G_b$. Thus, $S_1\cup S_2$ is an $L$-bounded
$s-t$ cut in $G$.  

In the second case, let $G'$ be the subgraph defined by the procedure in step
10; recall that $G'=G\setminus (B(b)\setminus\{s,t\})$.
As any $L$-bounded path in $G$ either intersects the set
$B(b)\setminus\{s,t\}$ or appears in $G'$, we conclude that
the set $S'\cup B(b)\setminus\{s,t\}$ is an $L$-bounded $s-t$ cut in $G$.

For the sake of completeness we also note that every recursive call
of the procedure is applied to a smaller graph, implying a 
finiteness of the algorithm; moreover, as the recursive calls in steps
6 and 7 are applied to almost disjoint subgraphs of $G$ (the two subgraphs
intersect in $s$ and $t$ only), the running time of the procedure is
polynomial in the size of the input graph.

Now we turn to the {\em performance} of the algorithm.  Let $cost(G)$ be the
cost of the solution of our algorithm and let $opt(G)$ be the size of the
optimal solution, for the graph $G$.  Similarly as for the proof of
correctness, we prove by induction on the recursion depth  
that $cost(G)\leq  \tau \cdot opt(G)$ where $\tau$ is the treewidth of $T$. 

We distinguish two cases of the final recursive call:
i) for subgraphs with no $L$-bounded $s-t$ path, dealt with in step 1,
the claim is obvious as $cost(G)=0$ in this case;
ii) for subgraphs such that the corresponding tree decomposition
does not contain any bag with both $s$ and $t$, dealt with in step 2,
note that for any node $c$ on the unique path connecting the disjoint
subtrees $T^s$ and $T^t$ of the decomposition tree $T$, the set 
$B(c)\setminus\{s,t\}$ is an $s-t$ cut of size at most $\tau$.

{\em Inductive step.}
We start with a simple lemma.
\begin{lemma}\label{lem:induction}
The $L$-bounded paths in $G_b$ are
internally vertex disjoint with the $L$-bounded paths in $G'$. 
\end{lemma}
\begin{proof}
The choice of the node $b$ in step 4 implies that every $L$-bounded $s-t$ path 
in $G_b$ uses at least one node from $B(b)\setminus\{s,t\}$.
Thus, there is no $L$-bounded
$s-t$ path in $G_b\setminus(B(b)\setminus\{s,t\})$.
As the vertex set $B(b)$ is a vertex cut (separator) in $G$ that disconnects
the subgraph $G_b\setminus B(b)$ of $G$ from the rest of the graph,
we conclude that every $s-t$ path in $G'$ is internally disjoint with every 
$s-t$ path in $G_b$.
\qed
\end{proof}

As in the proof of correctness, 
we distinguish the two cases $B(a)=\{s,t\}$ and $B(a)\not =\{s,t\}$.

In the first case, exploiting the fact that $B(a)=\{s,t\}$ is a separator in
$G$, we obtain $opt(G)=opt(G_b)+opt(\bar G_b)$. Plugging in the inductive
assumptions, $cost(G_b)\leq  \tau \cdot opt(G_b)$ and $cost(\bar G_b)\leq  \tau
\cdot opt(\bar G_b)$ and the earlier observed fact that $cost(G)\leq
cost(G_b)+cost(\bar G_b)$, we derive that $cost(G)\leq \tau \cdot opt(G)$.

In the second case,
from the description of the algorithm and the previous analysis
we know that $cost(G) \leq cost(G') +  \tau$. Exploiting the fact that by the choice of $b$
there is at least one
$L$-bounded $s-t$ path in $G_b$, Lemma~\ref{lem:induction} implies
$opt(G)\geq 1+opt(G')$.  Combining these observations with the 
inductive assumption $cost(G')\leq \tau \cdot opt(G')$, 
we obtain again $cost(G) \leq \tau \cdot opt(G)$.
This concludes the proof of the Theorem. \qed
\end{proof}

In the case that we are not given a tree decomposition on input, we just
start by computing one using one of the best known tree decomposition
algorithms.
Feige et al.~\cite{FHL:08} describe a polynomial time algorithm that
yields, for a given graph of treewidth $\tau$, a tree decomposition
with treewidth $\mathcal{O}(\tau\sqrt{\log \tau})$; for planar graphs and for graphs
excluding a fixed minor, the treewidth is $\mathcal{O}(\tau)$ only.
Similarly, for graphs with treewidth bounded by $\mathcal{O}(\log n)$,
Bodlaender et al.~\cite{BDD:16} describe how to find in polynomial time
a tree decomposition of treewidth $\mathcal{O}(\tau)$.
Depending on the input graph, one of these algorithms is used
to obtain a desired tree decomposition.
Thus, we obtain the following corollary.
\begin{corollary}\label{cor:approx}
There exists an $\mathcal{O}(\tau\sqrt{\log\tau})$-approximation
algorithm for the minimum $L$-bounded vertex cut on graphs with treewidth $\tau$;
for planar graphs, graphs excluding a fixed minor
and graphs with treewidth bounded by $\mathcal{O}(\log n)$,
there exists an $\mathcal{O}(\tau)$-approximation algorithm.
\end{corollary}


\section{Open problems}

A natural open problem for planar graphs is whether the shortest path most 
vital edges (vertices) problem is fixed parameter tractable on them, with respect
to the number $k$ of deleted edges (vertices). Despite the close relation
of the $L$-bounded cut problem and the shortest path most 
vital edges (vertices) problem, fixed parameter tractability of one of them
does not seem to easily imply fixed parameter tractability of the other problem.

The $\tau$-approximation for $L$-bounded vertex cuts is based
on the fact that bags in a tree decomposition yield 
{\em vertex} cuts of size at most equal the width of the decomposition.
Unfortunately, this is not the case for {\em edge} cuts -- one can
easily construct bounded treewidth graphs with no small balanced edge cuts.  Thus, 
another open problem is to look for  better approximation
algorithms for minimum $L$-bounded {\em edge} cuts, for graphs with treewidth
bounded by $\tau$.

Yet another challenging and more general open problem is to narrow the gap
between the upper and lower bounds on the approximation ratio of algorithms for
$L$-bounded cut for general graphs.

Finally, we note that the $L$-bounded edge cut problem in a graph $G=(V,E)$ 
is a kind of a {\em vertex ordering}
problem. We are looking for a mapping $\ell$ 
from the vertex set $V$ to the set $\{0,1,\ldots,L,L+1\}$
such that $\ell(s)=0$, $\ell(t)=L+1$ and the objective is to minimize 
the number of edges $\{u,v\}\in E$ for which $|\ell(u)-\ell(v)|>1$; 
given a feasible solution $F\subseteq E$,
the lengths of the shortest paths from $s$ to all other vertices in $G\setminus F$
yield such a mapping of cost $|F|$.
There are plenty of results dealing with {\em linear} vertex ordering problems where
one is looking for a bijective mapping from the vertex set $V$ to the set 
$\{1,2,\ldots,n\}$ minimizing some objective function. However, the requirement that
the mapping be bijective to a set of size $n$ seems crucial in the design and analysis of 
approximation algorithms for these problems. The question is whether it is possible to
obtain good approximations for some nontrivial {\em non-linear} vertex ordering problems.

\paragraph{Acknowledgments.} We thank Martin Kouteck\'y and Hans Raj Tiwary 
for stimulating discussions and an anonymous referee for suggestions 
yielding an improvement by a factor $\log n$ in the approximation
ratio in Theorem~\ref{thm:approx}.

\bibliographystyle{alpha}
\bibliography{lcuts}

\newcommand{\etalchar}[1]{$^{#1}$}
\begin{thebibliography}{BDD{\etalchar{+}}16}

\bibitem[AK71]{AK:71}
J.~Ad\'amek and V.~Koubek.
\newblock Remarks on flows in network with short paths.
\newblock {\em Commentationes Mathematicae Universitatis Carolinae},
  12(4):661--667, 1971.

\bibitem[BDD{\etalchar{+}}16]{BDD:16}
H.~L. Bodlaender, P.~G. Drange, M.~S. Dregi, F.~V. Fomin, D.~Lokshtanov, and
  M.~Pilipczuk.
\newblock A $c^k n$ $5$-approximation algorithm for treewidth.
\newblock {\em SIAM J. Comput}, 45(2):317--378, 2016.

\bibitem[BEH{\etalchar{+}}10]{BEH:10}
G.~Baier, T.~Erlebach, A.~Hall, E.~K{\"{o}}hler, P.~Kolman, O.~Pangr{\'{a}}c,
  H.~Schilling, and M.~Skutella.
\newblock Length-bounded cuts and flows.
\newblock {\em ACM Trans. Algorithms}, 7(1):4:1--4:27, 2010.
\newblock Preliminary version in Proc. of {ICALP} 2006.

\bibitem[BGV89]{BGV:89}
M.~O. Ball, B.~L. Golden, and R.~V. Vohra.
\newblock Finding the most vital arcs in a network.
\newblock {\em Operations Research Letters}, 8(2):73 -- 76, 1989.

\bibitem[BNKS95]{BKS:95}
A.~Bar-Noy, S.~Khuller, and B.~Schieber.
\newblock The complexity of finding most vital arcs and nodes.
\newblock Technical Report CS-TR-3539, Univ. of Maryland, Dept. of Computer
  Science, November 1995.

\bibitem[BNN15]{BNN:15}
C.~Bazgan, A.~Nichterlein, and R.~Niedermeier.
\newblock A refined complexity analysis of finding the most vital edges for
  undirected shortest paths.
\newblock In {\em Proc. of Algorithms and Complexity - 9th International
  Conference (CIAC)}, pages 47--60, 2015.

\bibitem[Bod88]{Bodlaender:88}
H.~L. Bodlaender.
\newblock Planar graphs with bounded treewidth.
\newblock Technical Report RUU-CS-88-14, Univ. Utrecht, Dept. of Computer
  Science, 1988.

\bibitem[CHLN05]{CHLN:05}
A.~Czumaj, M.~M. Halld\'orsson, A.~Lingas, and J.~Nilsson.
\newblock Approximation algorithms for optimization problems in graphs with
  superlogarithmic treewidth.
\newblock {\em Information Processing Letters}, 94(2):49--53, 2005.

\bibitem[DK15]{DK:15}
P.~Dvo\v{r}\'ak and D.~Knop.
\newblock Parametrized complexity of length-bounded cuts and multi-cuts.
\newblock In {\em Proc. of 12 Annual Conference on Theory and Applications of
  Models of Computation ({TAMC})}, pages 441--452, 2015.

\bibitem[Epp00]{Eppstein:00}
D.~Eppstein.
\newblock Diameter and treewidth in minor-closed graph families.
\newblock {\em Algorithmica}, 27(3):275--291, 2000.

\bibitem[FHL08]{FHL:08}
U.~Feige, M.~T. Hajiaghayi, and J.~R. Lee.
\newblock Improved approximation algorithms for minimum weight vertex
  separators.
\newblock {\em SIAM J. Comput}, 38(2):629--657, 2008.
\newblock Preliminary version in Proc. of {STOC} 2005.

\bibitem[FHNN15]{FHNN:15}
T.~Fluschnik, D.~Hermelin, A.~Nichterlein, and R.~Niedermeier.
\newblock Fractals for kernelization lower bounds, with an application to
  length-bounded cut problems.
\newblock {\em CoRR}, abs/1512.00333, 2015.

\bibitem[FHNN16]{FHNN:16}
T.~Fluschnik, D.~Hermelin, A.~Nichterlein, and R.~Niedermeier.
\newblock Fractals for kernelization lower bounds, with an application to
  length-bounded cut problems.
\newblock In {\em Proc. of 43rd International Colloquium on Automata,
  Languages, and Programming (ICALP)}, pages 25:1--25:14, 2016.

\bibitem[Fre90]{Freuder:90}
E.~C. Freuder.
\newblock Complexity of {$K$}-tree structured constraint satisfaction problems.
\newblock In {\em Proc. of the 8th National Conference on Artificial
  Intelligence}, pages 4--9, 1990.

\bibitem[GT11]{GT:11}
P.~A. Golovach and D.~M. Thilikos.
\newblock Paths of bounded length and their cuts: Parameterized complexity and
  algorithms.
\newblock {\em Discrete Optimization}, 8(1):72--86, 2011.

\bibitem[KBB{\etalchar{+}}08]{KBB:08}
L.~Khachiyan, E.~Boros, K.~Borys, K.~M. Elbassioni, V.~Gurvich, G.~Rudolf, and
  J.~Zhao.
\newblock On short paths interdiction problems: Total and node-wise limited
  interdiction.
\newblock {\em Theory Comput. Syst}, 43(2):204--233, 2008.

\bibitem[KK15]{KK:15}
P.~Kolman and M.~Kouteck{\'y}.
\newblock Extended formulation for {CSP} that is compact for instances of
  bounded treewidth.
\newblock {\em Electr. J. Comb}, 22(4):P4.30, 2015.

\bibitem[Klo94]{Kloks:94}
Ton Kloks.
\newblock {\em Treewidth: Computations and Approximations}, volume 842 of {\em
  Lecture Notes in Computer Science}.
\newblock Springer, 1994.

\bibitem[KRRS94]{KRRS:94}
P.~Klein, S.~Rao, M.~Rauch, and S.~Subramanian.
\newblock Faster shortest-path algorithms for planar graphs.
\newblock In {\em Proc. of the 26th ACM Symposium on Theory of Computing
  ({STOC})}, pages 27--37, 1994.

\bibitem[Lee17]{Lee:17}
E.~Lee.
\newblock Improved hardness for cut, interdiction, and firefighter problems.
\newblock In {\em Proc. of 44rd International Colloquium on Automata,
  Languages, and Programming (ICALP)}, pages 92:1--92:14, 2017.

\bibitem[LNP78]{LNP:78}
L.~Lov\'{a}sz, V.~{Neumann-Lara}, and M.~D. Plummer.
\newblock Mengerian theorems for paths of bounded length.
\newblock {\em Periodica Mathematica Hungarica}, 9:269--276, 1978.

\bibitem[MM10]{MM:10}
A.~R. Mahjoub and S.~T. McCormick.
\newblock Max flow and min cut with bounded-length paths: complexity,
  algorithms, and approximation.
\newblock {\em Math. Program}, 124(1-2):271--284, 2010.

\bibitem[PS16]{PS:16}
F.~Pan and A.~Schild.
\newblock Interdiction problems on planar graphs.
\newblock {\em Discrete Applied Mathematics}, 198:215--231, 2016.

\bibitem[RS84]{RS:84}
N.~Robertson and P.~D. Seymour.
\newblock Graph minors. {III.} planar tree-width.
\newblock {\em J. Comb. Theory, Ser. {B}}, 36(1):49--64, 1984.

\end{thebibliography}

\newpage
\begin{appendix}
\section{Appendix $L$-bounded Cut as a CSP Instance}
An instance $Q=(V,\DD,\HH,\CC)$ of CSP~\cite{KK:15}
consists of
\begin{itemize}
\item a set of {\em variables} $z_v$, one for each $v\in V$; without loss of
generality we assume that $V = \{1,\ldots,n\}$,
\item a set $\DD$ of finite {\em domains} $D_v\subseteq \RR$ (also
denoted $D(v)$), one for each $v\in V$,
\item a set of {\em hard constraints}
  $\HH \subseteq \{C_{U}\ |\ U \subseteq V \}$ where each
hard constraint $C_{U} \in \HH$ with $U=\{i_1, i_2,\dots,i_k\}$ and
$i_1 < i_2 < \cdots < i_k$, is a $|U|$-ary relation
$C_U \subseteq D_{i_1}\times D_{i_2}\times \cdots \times D_{i_k}$,
\item a set of {\em soft constraints}
  $\CC \subseteq \{C_U \ |\ U \subseteq V\}$ where each soft constraint
$C_{U} \in \CC$ with $U=\{i_1, i_2,\dots,i_k\}$ and
$i_1 < i_2 < \cdots < i_k$, is a $|U|$-ary relation
$C_{U}\subseteq D_{i_1}\times D_{i_2}\times \cdots \times D_{i_k}$.
\end{itemize}
%
For a vector $z^{}=(z^{}_1, z_2, \ldots,z^{}_n)$ and $U=\{i_1,
i_2,\dots,i_k\}\subseteq V$ with $i_1 < i_2 < \cdots < i_k$, we define
the {\em projection of} $z$ on $U$ as
$z^{}|_U=(z^{}_{i_1}, z_{i_2}, \ldots, z^{}_{i_k})$.
A vector $z\in \RR^n$ {\em satisfies the constraint} $C_U \in
\CC\cup \HH$ if and only if $z|_U \in C_U$.
We say that a vector
$z^{\star}=(z^{\star}_1,\ldots,z^{\star}_n)$ is {\em a feasible solution} for
$Q$ if $z^{\star} \in D_1\times D_2\times \ldots\times D_n$ and
$z^{\star}$ satisfies every hard constraint $C\in \HH$.
In the {\em maximization} ({\em minimization}, resp.) version of CSP, 
the
task is to find a {feasible solution} that maximizes (minimizes,
resp.) the number of \emph{satisfied} (unsatisfied, resp.) soft
constraints; the {\em cost} of a feasible solution is the number
of satisfied (unsatisfied, resp.) soft constraints.

The \emph{constraint graph} of $Q$ is defined as $H=(V,E)$
where $E= \{\{u,v\} \ |\  \exists C_{U} \in \CC\cup \HH \textrm{ s.t. } 
\{u,v\} \subseteq U\}$. We say that a {\em CSP instance $Q$ has bounded
treewidth} if the constraint graph of $Q$ has bounded treewidth.

Given an $L$-bounded edge cut instance $G=(V,E)$ with $s,t\in V$ and an integer~$L$, 
we construct the corresponding minimization CSP instance
$Q=(V,\DD,\HH,\CC)$ as follows. 
The set of variables of $Q$ coincides with the set $V$ of vertices of $G$
and for each $v\in V$, the corresponding domain $D_v$ 
is $\{0,1,\ldots,L\}$. The set of hard constraints $\HH$ consists
of two constraints, $C_{\{s\}}=\{0\}$ and $C_{\{t\}}=\{L\}$. The set
of soft constraints $\CC$ contains a constraint $$C_{(i,j)}=\{(0,0),(0,1),(L,L-1),(L,L)\}
\cup \bigcup_{i=1}^{L-1}\{(i,i-1),(i,i),(i,i+1)\}$$ for each edge
$\{i,j\}\in E$, $i<j$, of the graph $G$.

To see that a feasible solution for the constructed instance $Q$ of CSP
corresponds to a feasible solution of the $L$-bounded cut problem of the same
cost, and vice versa, we observe the following. 

Given an optimal solution $F\subset E$ of the $L$-bounded cut problem with $s$
and $t$ in the same component, the vector of shortest path distances from $s$
to all other vertices in $(V,E\setminus F)$ yields a feasible solution for the
CSP instance $Q$ (to be more precise, if some of the distances are larger than $L$,
we replace in the vector every such value by $L$); given an optimal solution $F\subset E$ of the $L$-bounded cut
problem with $s$ and $t$ in the different components, we obtain a feasible
solution for $Q$ by assigning the value $0$ to every vertex in the
$s$--component and the value $2$ to every vertex in the $t$--component. 
Note that in both cases the cost of the $L$-bounded cut and the cost of the
$CSP$ instance $Q$ are the same. On the other hand, for every feasible solution 	
$(z_1,\ldots,z_{n})$ of the instance $Q$, the set $F=\{\{u,v\}\ \mid \ 
|z_u-z_v|>1\}$ is an $L$-bounded cut of the same cost.
Finally, we note that the constraint graph of $Q$ coincides with the original graph $G$.

For the vertex version of the $L$-bounded cut problem, we extend the domain of every vertex
$v$ by an extra element $-1$ representing the fact that $v$ belongs to the $L$-bounded
cut, and we adjust the hard and soft constraints appropriately.

\end{appendix}

\end{document}